\documentclass[12pt]{article}
\usepackage{amssymb}
\usepackage{amsmath}
\usepackage{graphicx,psfrag,epsf}
\usepackage{enumerate}
\usepackage{natbib}
\usepackage{psfrag,epsf}
\usepackage{color}
\usepackage{algorithm}
\usepackage[noend]{algpseudocode}
\usepackage{mathtools}
\mathtoolsset{showonlyrefs}
\usepackage{blindtext}
\usepackage{flowchart}
\usepackage{thmtools}
\usepackage{thm-restate}
\usepackage{tikz}
\usepackage{array}
\usepackage{caption}
\usepackage{graphicx}
\usepackage{float}
\usepackage{subcaption}
\usepackage{amsthm}
\usepackage[skip=0.1pt]{caption}
\allowdisplaybreaks
\newcommand{\blind}{0}
\newtheorem{assumption}{Assumption}
\newtheorem{theorem}{Theorem}

\newtheorem{lemma}{Lemma}

\newcommand{\+}{\mathbf}
\newcommand{\bs}{\boldsymbol}

\newcommand{\E}{\mathbb{E}}
\renewcommand{\circ}{ \odot}

\begin{document}

\def\spacingset#1{\renewcommand{\baselinestretch}%
{#1}\small\normalsize} \spacingset{1}


\if0\blind
{
  \title{\bf Accelerated inference for stochastic compartmental models with over-dispersed 
 partial observations}
  \author{Michael Whitehouse\hspace{.2cm}\\
    School of Public Health, Imperial College}
  \maketitle
} \fi

\if1\blind
{
  \bigskip
  \bigskip
  \bigskip
  \begin{center}
    {\LARGE\bf Title}
\end{center}
  \medskip
} \fi

\bigskip
\begin{abstract}
An assumed density approximate likelihood is derived for a class of partially observed stochastic compartmental models which permit observational over-dispersion. This is achieved by treating time-varying reporting probabilities as latent variables and integrating them out using Laplace approximations within Poisson Approximate Likelihoods (LawPAL), resulting in a fast deterministic approximation to the marginal likelihood and filtering distributions. We derive an asymptotically exact filtering result in the large population regime, demonstrating the approximation's ability to recover latent disease states and reporting probabilities. Through simulations we: 1) demonstrate favorable behavior of the maximum approximate likelihood estimator in the large population and time horizon regime in terms of ground truth recovery; 2) demonstrate order of magnitude computational speed gains over a sequential Monte Carlo likelihood based approach and explore the statistical compromises our approximation implicitly makes. We conclude by embedding our methodology within the probabilistic programming language Stan for automated Bayesian inference to develop a model of practical interest using data from the Covid-19 outbreak in Switzerland.
\end{abstract}

\noindent%

\spacingset{1.45}
\section{Introduction}
\label{sec:intro}

Compartmental models, initially developed in the seminal works of \cite{m1925applications, kermack1927contribution}, capture disease outbreak dynamics by dividing a population into distinct disease states we call `compartments'. One then defines a transmission model which describes the rate at which individuals transition between compartments through time - in turn, this implicitly characterizes the progression of the disease through the population. In practice this process is noisily observed through testing data, usually cast as random under-reporting of either incidence (\textit{newly} infected individuals) or prevalence (\textit{currently} infected individuals) counts.

Such models can be cast in myriad ways, deterministic, stochastic, continuous, or discrete time, finite population counts, or continuous proportions. The link between the Markov jump process formulation and its corresponding large population limiting deterministic ODE system was formalized by \cite{kurtz1970solutions}. When coupled with a noisy observation model, likelihoods for ODE transmission models are relatively easy to compute using an ODE solver. This has motivated widespread adoption of the ODE approach in practice. Such an approach, however, ignores the inherent stochasticity governing disease transmission \citep{king2015avoidable}.

A complication introduced when using a stochastic transmission model is the need to marginalize over the latent transmission process, which subsequently renders the likelihood intractable. Many simulation based solutions to this issue have been proposed in the literature, approaches include: sequential Monte Carlo (SMC) \citep{corbella2022lifebelt, koepke2016predictive,wheeler2024informing}, Approximate Bayesian Computation \citep{mckinley2009inference, kypraios2017tutorial}, Data Augmentation MCMC \citep{fintzi2017efficient,lekone2006statistical,nguyen2021stochastic}. In principle these methods only require the ability to simulate from the model, through in practice they rely on intricate tuning, vast computational power for many repeated simulations, and the construction of sophisticated proposal distributions or summary statistics \citep{rimella2022approximating,prangle2018summary}. One deterministic approach is the linear noise approximation (LNA) \citep{golightly2023accelerating} which consists of Gaussian approximations derived from the functional central limit theorem associated with \cite{kurtz1970solutions}.

The recently developed suite of assumed density approximate likelihoods (ADAL) \citep{whiteley2021inference, whitehouse2023consistent,rimella2025scalable} compose a family of computationally efficient methods for inference in discrete time stochastic compartmental models which exhibit favorable theoretical properties. ADALs are obtained by marginalizing over the latent compartmental model through a combination of recursive expectation propagation \citep{minka2001family} and assumed density \citep{sorenson1968non} steps to arrive at deterministic approximations to filtering distributions and the marginal likelihood. In their vanilla form they can be used to fit models with Binomial under-reporting \citep{whiteley2021inference}, this is an equi-dispersed observation model in the sense that the variance of a binomial random variable is strictly of the same order as its mean. Such assumptions can lead to problematic inference when confronted with data exhibiting over-dispersion, i.e. an under-reporting distribution with variance significantly larger than its mean, as explored in Section \ref{sec:covid} and noted in previous studies \citep{whitehouse2023consistent,stocks2020model}. The use of ADALs can be extended to over-dispersed observation mechanisms by treating reporting probabilities as time varying random variables, see Section \ref{sec:obsmod} for details, and integrating them out by embedding ADALs in an SMC sampler, incurring significant computational cost \citep{whitehouse2023consistent}. In this paper we extend the remit of ADALs to models permitting observational over-dispersion without resorting to sampling methods. This is achieved by recursively using Laplace approximations to integrate out the latent reporting rates, resulting in a novel deterministic and computationally simple marginal likelihood approximation.
We introduce the class of latent compartmental models considered by this paper in Section \ref{sec:model}. We derive and present an algorithm to compute our approximate filtering distributions and likelihood in Section \ref{sec:recursions}. In Section \ref{sec:theory} we present a result on the asymptotic behavior of our filtering approximations in the large population limit. Through simulations in Section \ref{sec:examples} we explore properties of the maximum approximate likelihood estimator in terms of ground truth recovery and benchmark our approximation against an SMC approach. To conclude we perform a comparison of observation mechanisms by embedding our methodology within a probabilistic programming language (PPL) to facilitate automated Bayesian inference.
\vspace{-0.5cm}
\section{Model}
\label{sec:model}
\subsection{Notation}
The set of natural numbers, including $0$, is denoted $\mathbb{N}_0$. The set of non-negative real numbers is denoted $\mathbb{R}_{\geq 0}$. For an integer $m\geq 1$, $[m]\coloneqq\{1,\ldots,m\}$. Matrices and vectors are denoted by bold upper-case and bold lower-case letters, respectively, e.g., $\+{A}$ and $\+{b}$, with non-bold upper-case and lower case used for their respective elements $A^{(i,j)}$, $b^{(i)}$. All vectors are column vectors unless stated otherwise. We use $\+1_m$ to denote the vector of $m$ $1$'s. The indicator function is denoted $\mathbb{I}[\cdot]$. For $\+ x \in \mathbb{N}_0^m$ we define $\bs \eta(\+x) = [x^{(1)}/\+1_m^\top \+x \,\cdots\, x^{(m)}/\+1_m^\top \+x]^\top$ if $\+1_m^\top\+x>0$, i.e. $\bs \eta(\+x)$ normalizes $\+x$ to yield a probability vector;  and $\bs \eta(\+x) = \bs 0_m$ if $\+1_m^\top\+x=0$.


For $\+x\in \mathbb{N}_0^m$ and $\bs \lambda \in \mathbb{R}_{\geq 0}^m$ we write $\+x \sim \mathrm{Pois}(\bs \lambda)$ to denote that
the elements of $\+x$ are independent and element $x^{(i)}$ is Poisson distributed with parameter $\lambda^{(i)}$. We shall say that such a
 random vector $\+x$ has a ``vector-Poisson distribution''. For a probability vector $\bs
 \pi$  we write $\mathrm{Mult}(n,\bs\pi)$ for the associated multinomial distribution.
Similarly, for a random matrix $\+X \in \mathbb{N}_0^{m \times l}$ and a matrix $\bs \Lambda \in \mathbb{R}_{\geq 0}^{m \times l}$, we write $\+X\sim \mathrm{Pois}(\bs \Lambda) $ when the elements of $\+X$ are independent with $X^{(i,j)}$ being Poisson distributed with parameter $\Lambda^{(i,j)}$.
We call $\bs \lambda$ (resp. $\bs\Lambda$) the intensity vector (resp. matrix). For $a,b,\mu\in \mathbb{R}$, $a<b$ and $\sigma^2>0$ we write $\mathcal{N}_{[a,b]}(\mu,\sigma^2)$ to denote a normal random variable with mean and variance parameters $\mu$ and $\sigma^2$ which is restricted to the support $[a,b]$.

\subsection{Latent compartmental model} \label{sec:lat_comp_model}
The transmission model we consider for the remainder of the paper is defined by: $m$, the number of compartments; $n$ the initial population size; a length-$m$ probability vector $\bs\pi_0$ and for each $t \geq 0$  a mapping from length-$m$ probability vectors to size-$m \times m$ row-stochastic matrices, $\boldsymbol{\eta} \mapsto \mathbf{K}_{t, \boldsymbol{\eta}}$. The population at time $t \in \mathbb{N}_0$ is a set of $n$ of random variables $\{\xi_{t}^{(1)}, \ldots, \xi_{t}^{(n)}\}$, each valued in $[m]$.  The counts of individuals in each of the $m$ compartments at time $t$ are
collected in $\mathbf{x}_{t}=[x_{t}^{(1)} \cdots x_{t}^{(m)}]^\top$, where $x_{t}^{(i)}=\sum_{j=1}^{n} \mathbb{I}[\xi_{t}^{(j)}=i].$ The population is initialized as a draw ${\+x}_0 \sim \mathrm{Mult}(n, \bs \pi_0)$. The members of the population are exchangeable, labeled by, e.g., a uniformly random assignment of indices $\{\xi_{0}^{(1)}, \ldots, \xi_{0}^{(n)}\}$ subject to $x_{0}^{(j)}:=\sum_{i=1}^{n} \mathbb{I}[\xi_{0}^{(i)}=j].$ For $t\geq 1$, given $\{\xi_{t-1}^{(1)}, \ldots, \xi_{t-1}^{(n)}\}$, the $\{\xi_{t}^{(1)}, \ldots, \xi_{t}^{(n)}\}$ are conditionally independent such that for each $i$, $\xi_{t}^{(i)}$ is drawn from the $\xi_{t-1}^{(i)}$'th row of  $\mathbf{K}_{t, \boldsymbol{\eta}({\+ x}_{t-1})}$. Moreover, let $\mathbf{Z}_{t}$ be an $m \times m$ matrix with elements $Z_{t}^{(i, j)}:=\sum_{k=1}^{n} \mathbb{I}[\xi_{t-1}^{(k)}=i, \xi_{t}^{(k)}=j],$ which counts the individuals transitioning from
compartment $i$ at $t-1$ to $j$ at time $t$. It should be understood that the way in which $K_{\bs \eta}$ depends on $\bs \eta$ encodes how transition rates depend on the current population disease states (e.g. prevalence), we illustrate this now with an example.



\subsubsection{SEIR example} \label{example:SEIR}
A popular model used in practice \citep{lekone2006statistical,rawson2025mathematical} is the discrete time stochastic susceptible exposed infected removed ($SEIR$) model. The dynamics of susceptible $(S_t)$, exposed $(E_t)$, infected $(I_t)$, and  removed $(R_t)$ counts are described by
\begin{equation*}
S_{t+1} = S_{t} - B_{t},\quad E_{t+1} = E_{t} + B_{t} - C_{t},\quad I_{t+1} = I_{t} + C_{t} - D_{t}, \qquad R_{t+1} = R_{t} + D_{t}.
\end{equation*}
With conditionally independent, binomially distributed random increment variables:
\begin{equation*}
B_t \sim \mathrm{Bin}(S_t, 1- e^{-h\beta\frac{I_t}{n}}), \quad C_t \sim \mathrm{Bin}(E_t, 1- e^{-h\rho}),\quad D_t \sim \mathrm{Bin}(I_t, 1- e^{-h\gamma}),
\end{equation*}
where $h>0$ is a time-step size (set $h=1$ unless stated otherwise) and $\beta, \rho, \gamma$ are model parameters. The model is initialized with non negative integers in each compartment as a sample $[S_0,E_0,I_0,R_0]^\top \sim\mathrm{Mult}(n,\bs\pi_0)$ for some length-$4$ probability vector $\bs\pi_0$. We can interpret $\beta$ as the rate at which interactions between susceptible and infected individuals results in an infection. The mean exposure and infection periods are given by $1/\rho$ and $1/\gamma$ respectively and the reproduction number is given by $\beta/\gamma$. This model can be cast as an instance of the model from Section \ref{sec:lat_comp_model} by taking $m=4$, identifying $\+{x}_t\equiv[S_t\;E_t\;I_t\;R_t]^\top$ and:
\begin{equation}\label{eq:SEIRexample}
\mathbf{K}_{t, \boldsymbol{\eta}}=\left[\begin{array}{cccc}
e^{-h \beta \eta^{(3)}} & 1-e^{-h \beta \eta^{(3)}} & 0 &0\\
0 &  e^{-h \rho} & 1-e^{-h \rho} &0\\
0 &  0&e^{-h \gamma} & 1-e^{-h \gamma}\\
0 & 0  &0 &1
\end{array}\right].
\end{equation}
We can then further identify the time $t$ transitions $B_t\coloneqq Z_t^{(1,2)}$, $C_t\coloneqq Z_t^{(2,3)}$, and $D_t\coloneqq Z_t^{(3,4)}$.

\subsection{Over-dispersed observation model of disease incidence} \label{sec:obsmod}
In this article we neglect observation models derived from disease prevalence, the total number currently infected in a population at a given time, since these data are rarely available in practice. In our notation this corresponds to observations derived from the process $\{\+x_t\}_{t\geq 0}$. Instead, we focus on observations derived from disease incidence, the number of \textit{newly} infected individuals in a given time window - corresponding to observations derived from the process $\{\+Z_t\}_{t\geq 1}$. In particular, the number of individuals transitioning between between two compartments $i,j \in [m]$ at time $t\geq 1$ is modeled by $Z_t^{(i,j)}$. For pedagogical reasons we consider here incidence corresponding to a single pair of compartments $i$ and $j$, i.e. observations are a univariate time series $y_{1},y_2,\dots$, though we note that the methodology could be extended to higher dimensional problems, such as age-stratified incidence or geographical meta-population models incorporating immigration and emigration \citep{ionides2022iterated, andrade2020evaluation}, we discuss this further in Section \ref{sec:discussion}. The observation at time $t$, denoted $y_t$, has  distribution given by the following hierarchical model.
\begin{align}
    q_t &\sim f(\cdot \mid \varphi), \label{eq:qprior}\\
    y_t &\sim \text{Bin}(Z_t^{(i,j)},q_t),
\end{align}
where $f(\cdot\mid\varphi)$ is a distribution with support $[0,1]$ and $\varphi$ is a model parameter. If $f(\cdot | \varphi)$ is a Dirac mass at some $\mu_q \in [0,1]$ then we recover an equi-dispersed binomial observation model \citep{whiteley2021inference,fintzi2017efficient}. Considering a general $f$ gives rise to over-dispersed models for $y_t$ following when $q_t$ is marginalized out. For the remainder of the paper we follow \citep{whitehouse2023consistent} and take $f$ to be a Gaussian distribution truncated to the interval $(0,1)$ with parameter $\varphi \coloneqq [\mu_q ,\sigma_q^2]$, where the components denote mean and variance parameters, respectively. In Section \ref{sec:recursions} we show that with this model structure it is possible to derive an analytically tractable assumed density approximate likelihood.
\vspace{-0.5cm}
\section{Approximate filtering and marginal likelihood}\label{sec:recursions}

This section presents methodology for performing inference on the model described in Section \ref{sec:model}; the proposed approximations fall under the umbrella of assumed density approximate likelihood (ADAL) methods \citep{whiteley2021inference,rimella2025scalable,whitehouse2023consistent}, extending previous advances to handle over-dispersed observations.



\subsection{Model likelihood}
The processes $\{(\+Z_t,q_t)\}_{t\geq1}$ and $\{y_t\}_{t\geq1}$ can be cast as a hidden Markov model (HMM), with the former constituting the latent hidden state and latter the observed. Given a finite time horizon $T\in \mathbb{N}$ the exact likelihood is given by
\begin{equation}
        p(y_{1:T}) = \sum_{\+x_0}p(\+x_0)\sum_{\+Z_{1:T}}\int_{(0,1)^{T}}\prod_{t=1}^T p(y_t|\+Z_t,q_t) p(\+Z_t|\+Z_{t-1})p(q_t)dq_{1:T},
\end{equation}
which has prohibitive computational cost for even moderate population size $n$, compartment number $m$ and time horizon $T$ \citep{whiteley2021inference}. The forward algorithm reduces this cost to linear in $T$ using the so-called prediction and update steps, in principle, given $p(\+x_0)$ one could compute
\begin{align}
    p(\+Z_t| y_{1:t-1}) &= \sum_{\+Z_{t-1}}p(\+Z_t|\+Z_{t-1})p(\+Z_{t-1}|y_{1:t-1}), \label{eq:predZ} \\
    p(y_t|y_{1:t-1}) &= \sum_{\+Z_t} \int p(y_t| \+Z_t,q_t)p(\+Z_t,q_t| y_{1:t-1})d{q_t}, \label{eq:likinc} \\
    p(q_t | y_{1:t}) &= \frac{p(y_t|q_t)p(q_t|y_{1:t-1})}{p(y_t|y_{1:t-1})}, \label{eq:qupdate}\\
    p(\+Z_t| y_{1:t}) &= \frac{p(y_t|\+Z_t)p(\+Z_t|y_{1:t-1})}{p(y_t|y_{1:t-1})}. \label{eq:Zuptade}\\
\end{align}
However, equations \eqref{eq:predZ}-\eqref{eq:Zuptade} are still not available in closed form due to the intractable marginalization steps over $\+Z_{t-1},\+Z_t,$ $q_t$, and must be approximated.
\subsection{Approximation derivation}
Our strategy to approximate \eqref{eq:predZ}-\eqref{eq:Zuptade} is outlined by the following approximation steps:
\begin{enumerate}
    \item Given a matrix Poisson approximation $\mathrm{Poi}(\bar{\+\Lambda}_{t-1})$ of $p(\+Z_{t-1}|y_{1:t-1})$ for some $m\times m$ matrix $\bar{\+\Lambda}_{t-1}$, derive an approximation $\mathrm{Poi}({\+\Lambda}_{t})$ of $p(\+Z_t| y_{1:t-1})$ for some $m\times m$ matrix ${\+\Lambda}_{t}$.
    \item Perform a Laplace approximation to the intractable calculations \eqref{eq:likinc} and \eqref{eq:qupdate} to approximate $p(y_t|y_{1:t-1})$ and $p(q_t | y_{1:t})$.
    \item Calculate an approximation of the first moment  $\mathbb{E}[\+Z_t| y_{1:t}]$ and use this to define an approximation $\mathrm{Poi}(\bar{\+\Lambda}_{t})$ of $p(\+Z_t| y_{1:t})$ for some $m\times m$ matrix $\bar{\+\Lambda}_{t}$.
\end{enumerate}
These steps define a recursive strategy to approximate \eqref{eq:predZ}-\eqref{eq:Zuptade}. Step 1 follows the same reasoning as the prediction step associated with the Poisson Approximate Likelihood (PAL) procedure of \cite{whitehouse2023consistent}. The recipe for steps 2 and 3, along with supporting theory and experiments in Sections \ref{sec:theory} and \ref{sec:examples}, constitute the novel contributions of this paper.
\paragraph{Approximation step 1.}
Step 1 can be formalized with the following Lemma, which was first presented in \cite{whitehouse2023consistent}. For $\+Z\in\mathbb{N}_0^{m\times m}$ and a length-$m$ probability vector $\bs \eta$, let $\bar{M}_t(\+Z, \bs \eta, \cdot)$ denote the probability mass function of a random $m \times m$ matrix, say $\tilde{\+ Z}$, such that $\+1_m^\top \+ Z = (\tilde{\+ Z} \+1_m )^\top$ with probability $1$ where, given the row sums $\tilde{\+Z}\+1_m = \+ x$, the rows of $\tilde{\+Z}$ are conditionally independent with the conditional distribution of the $i^{th}$ row being $\mathrm{Mult}(x^{(i)}, \+K_{t,\bs \eta}^{(i,\cdot)})$. That is, by construction, we have that $\bar{M}_t(\+Z_{t-1}, \bs \eta(\+1_m^\top \+Z_{t-1}), \+Z_t)$ is equal to $p(\+Z_{t}|\+Z_{t-1})$ as per our latent compartmental model.

\begin{lemma}\label{lem:Zpred}
If for a given $m \times m$ matrix $\+\Lambda$, $\bar \mu$ is the probability mass function associated with $\mathrm{Pois}(\+ \Lambda)$ and $\E_{\bar \mu}[ \+1_m^\top\+Z]$ is the expected value of $\+1^\top_m \+Z$ where $\+Z \sim \bar \mu$, then we have that $\sum_{\+Z\in \mathbb{N}_{0}^{m\times m}}\bar \mu(\+Z) \bar M_t(\+Z,\bs \eta(\E_{\bar \mu}[\+1_m^\top\+Z]),\cdot)$ is the probability mass function associated with $\mathrm{Pois}((\bs \lambda \otimes \+1_m)\circ \+K_{t,\bs \eta(\bs \lambda)})$, where $\bs \lambda^\top \coloneqq \+ 1_m^\top \bs \Lambda$.
\end{lemma}
Using Lemma \ref{lem:Zpred} we arrive the following recipe for approximation step 1. Given a matrix Poisson approximation $\mathrm{Poi}(\bar{\+\Lambda}_{t-1})$ of $p(\+Z_{t-1}|y_{1:t-1})$ for some $m\times m$ matrix $\bar{\+\Lambda}_{t-1}$, we approximate $p(\+Z_t| y_{1:t-1})$ with $\mathrm{Poi}({\+\Lambda}_{t})$ where $\+\Lambda_{t} = (\bar{\bs \lambda}_{t-1} \otimes \+1_m)\circ \+K_{t,\bs \eta(\bar{\bs \lambda}_{t-1})}$ where $\bar{\bs \lambda}_{t-1}^\top\coloneqq \+1_m^\top\bar{\+\Lambda}_{t-1}$.
\paragraph{Approximation step 2: marginal likelihood.}

 Given our approximation $\mathrm{Poi}({\+\Lambda}_{t})$ of $p(\+Z_t| y_{1:t-1})$ we wish to perform a Laplace approximation to equations \eqref{eq:likinc} and \eqref{eq:qupdate}, we now outline this recipe. Once again we distill the problem to suppress dependence on $t$ and $y_{1:{t-1}}$: for some $m\times m$ matrix $\+\Lambda$, $(i,j)\in [m]^2$, $\mu_q \in (0,1)$, and $\sigma_q^2>0$, let $\+Z \sim \mathrm{Poi}(\+\Lambda)$, $q \sim \mathcal{N}_{[0,1]}(\mu_q , \sigma_q^2)$, and $y \sim \mathrm{Binom}(Z^{(i,j)},q)$. One can use standard results of Poisson processes \citep{kingman1992poisson} to notice that, given $q$, $y$ is a marked Poisson process and we can write $y|q\sim \mathrm{Poi}(q\Lambda^{(i,j)})$. The standard recipe for a Laplace approximation, see e.g.  \cite{murphy2012probabilistic} Section 8.4.1, to $p(q|y)$ and $p(y)$ is given by
\begin{align}
    p(q\mid y) &\overset{\text{\tiny {\textit{law}}}}{\approx } \mathcal{N}_{[0,1]}(\bar q, s^2), \label{eq:qfilt}\\
    \log p(y) &\approx \log p(y, \bar q)) + \log(s\sqrt{2\pi}), \label{eq:liky}
\end{align}
where $\bar q := \arg\max_qp(y,q)$ and $s^2 := \left(-\frac{\partial^2}{\partial q^2}\log p(y,q)|_{q=\bar q} \right)^{-1}$ where the support in \eqref{eq:qfilt} has further been restricted to the unit interval $[0,1]$ in line with the prior support on $q$. It turns out both $\bar q$ and $s^2$ can be calculated analytically: for some constant $C_y $, which does not depend on $q$, write
\begin{equation}
    \log p(y,q)=\log \left\{p(y|q)p(q) \right\} = y\log(q\Lambda^{(i,j)}) - q\Lambda^{(i,j)} - \frac{1}{2}\left(\frac{q-\mu_q}{\sigma^2_q} \right)^2+ C_y.
\end{equation}
Differentiating and setting to zero gives a quadratic equation we can solve
\begin{align}
    \frac{\partial}{\partial q}\log p(y,q) = \frac{y}{q} - \Lambda^{(i,j)} - \frac{q-\mu_q}{\sigma^2_q} &= 0, \\
   \iff q^2 +( \Lambda^{(i,j)}\sigma^2_q - \mu_q)q - y\sigma^2_q &= 0,
\end{align}
which gives $\arg\max_{q}p(y,q)= \frac{1}{2}\left( \mu_q - \Lambda ^{(i,j)}\sigma^2_q + \sqrt{\left( \Lambda ^{(i,j)}\sigma^2_q  - \mu_q \right)^2 - 4y\sigma^2_q }\right),$ which we can in turn use to deduce  $\left(-\frac{\partial^2}{\partial q^2}\log p(y,q)|_{q=\bar q} \right)^{-1} =  \left( \frac{y}{\bar{q}^2}+ \frac{1}{\sigma_q^2}\right)^{-1}$, where we take the convention $0/0=0$. If we let $\phi(\cdot)$ be the probability density function associated with a $\mathcal{N}_{[0,1]}(\mu_q , \sigma_q^2)$ random variable then we can write \eqref{eq:liky} as
\begin{align}
     \log p(y) &\approx \log \left\{p(y|q)p(q) \right\} + \log(s\sqrt{2\pi})\\
     &=y\log(\bar q\Lambda^{(i,j)}) - \bar q\Lambda^{(i,j)} - y! + \phi(\bar q) + \log(s\sqrt{2\pi}). \label{eq:marglik}
\end{align}
Our recipe for approximation step 2 is thus as follows. Given an approximation $\mathrm{Poi}({\+\Lambda}_{t})$ of $p(\+Z_t| y_{1:t-1})$, we define $\bar{q}_t \coloneqq \frac{1}{2}\left( \mu_q - \Lambda_t ^{(i,j)}\sigma^2_q + \sqrt{\left( \Lambda_t ^{(i,j)}\sigma^2_q  - \mu_q \right)^2 - 4y_t\sigma^2_q }\right)$ and $s^2_t \leftarrow \left( \frac{y_t}{\bar{q}_t}+ \frac{1}{\sigma_q^2}\right)^{-1}$ to arrive at the approximation:

\begin{equation}
   \log p(y_t|y_{1:t-1}) \approx \ell(y_t|y_{1:t-1}) \coloneqq - \Lambda_t^{(i,j)} \bar{q}_t+ + y_t \log(\Lambda_t^{(i,j)} \bar{q}_t) - y_t! + \phi\left(\bar{q}_t\right) +\frac{1}{2}\log(s_t2\pi)
\end{equation}

\paragraph{Approximation step 3: update step.}
To complete our recursions we move on to approximation step 3 in which we use a moment matching strategy to approximate \eqref{eq:Zuptade} with an assumed matrix Poisson distribution. For any $(k,l)\neq (i,j)$ we have $\mathbb{E}\left[Z^{(k,l)}|y\right]=\mathbb{E}\left[Z^{(k,l)}\right]=\Lambda^{(k,l)}$. Considering now $Z^{(i,j)}$, we have by Lemma 4 of \cite{whitehouse2023consistent} that
\begin{equation}
    \mathbb{E}\left[Z^{(i,j)}|q,y\right] = y + (1-q)\Lambda^{(i,j)}.
\end{equation}
Then by the tower rule we have that
\begin{equation}
    \mathbb{E}\left[Z^{(i,j)}|y\right] =\mathbb{E}\left[\mathbb{E}\left[Z^{(i,j)}|q,y\right]|y\right] = y + (1-\mathbb{E}\left[q|y\right])\Lambda^{(i,j)}.
\end{equation}
We then approximate $\mathbb{E}\left[q|y\right]$ using \eqref{eq:qfilt}. We make the further approximation that $\mathbb{E}\left[q|y\right] \approx \bar q$, this error becomes negligible when the variance of \eqref{eq:qfilt} is small and can be justified both theoretically and experimentally as per results presented in Sections \ref{sec:theory} and \ref{sec:examples}.  Defining the $m \times m$ matrix $\bar {\+\Lambda}$ element-wise as $\bar {\Lambda}^{(k,l)} := {\Lambda}^{(k,l)}$  for  $(k,l)\neq (i,j)$ and $\bar {\Lambda}^{(i,j)} := y + (1-\bar q){\Lambda}^{(i,j)}$ we have that $\mathbb{E}\left[\+Z|y\right] \approx \bar{\+\Lambda}$, leading us to the moment matching 
 approximation of $p(\+Z|y)$ with $\mathrm{Poi}(\bar{\+\Lambda})$. Hence our recipe for approximation step 3 is as follows. Given the definitions of $\+\Lambda_{t}$ and $\bar q_t$ from steps 1 and 2, we approximate $p(\+Z_t| y_{1:t})$ with $\mathrm{Poi}(\bar{\+\Lambda}_{t})$ where
 \begin{equation}
     \bar{\Lambda}^{(k,l)}_{t} = \begin{cases}
         {\Lambda}_{t}^{(k,l)} \quad &\text{ if } {(k,l)} \neq {(i,j)}, \\
         y_t + (1-\bar q_t){\Lambda_t}^{(k,l)}  \quad &\text{ if } {(k,l)} = {(i,j)}.
     \end{cases}
 \end{equation}
 
\subsection{Laplace approximations within PALs algorithm}
We can now collect the approximation steps and present them and present them as a recursive algorithm, which we call the Laplace approximations within Poisson Approximate Likelihoods (LawPAL) algorithm. Line \ref{algline:step1} is informed by approximation step 1, lines \label{algline:step2a}-\label{algline:step2b} by step 2, and lines \label{algline:step3a} and \label{algline:step3b} by step 3.
\begin{algorithm}[H]
\caption{LawPAL}\label{alg:Z}
\begin{algorithmic}[1]
  \Statex {\bf initialize:} $\bar{\bs \lambda}_{0} \leftarrow \bs \lambda_0:=n\bs\pi_0$
  \State {\bf for}  $t \geq 1$:
  \State \quad $\+ \Lambda_{t} \leftarrow (\bar{\bs \lambda}_{t-1} \otimes \+1_m)\circ \+K_{t,\bs \eta(\bar{\bs \lambda}_{t-1})}$ \label{algline:step1}
  \State \quad $\bar{q}_t \leftarrow \frac{1}{2}\left( \mu_q - \Lambda_t ^{(i,j)}\sigma^2_q + \sqrt{\left( \Lambda_t ^{(i,j)}\sigma^2_q  - \mu_q \right)^2 - 4y_t\sigma^2_q }\right)$ \label{algline:step2a}
  \State \quad $s^2_t \leftarrow \left( \frac{y_t}{\bar{q}_t}+ \frac{1}{\sigma_q^2}\right)^{-1}$  
  \State \quad $\ell(y_t|y_{1:t-1}) \leftarrow - \Lambda_t^{(i,j)} \bar{q}_t+ + y_t \log(\Lambda_t^{(i,j)} \bar{q}_t) - y_t! + \phi\left(\bar{q}_t\right) +\frac{1}{2}\log(s_t2\pi)$ \label{algline:step2a}
   \State \quad $\bar{\+ \Lambda}_t \leftarrow \bs \Lambda_{t} $  \label{algline:step3a}
  \State \quad $\bar{\Lambda}_t^{(i,j)} \leftarrow y_t+ (1- \bar{q}_t  )  \Lambda_{t}^{(i,j)} $  \label{algline:step3b}
  \State \quad$\bar{\bs \lambda}_t \leftarrow (\+ 1_m^\top \bar{\+ \Lambda}_t)^\top$
  \State {\bf end for}
\end{algorithmic}
\end{algorithm}
As output from algorithm \ref{alg:Z} one may take the approximate filtering distributions $ p(q_t | y_{1:t}) \approx \mathcal{N}_{[0,1]}(\bar q_t,s_t^2)$, $p(\+x_t| y_{1:t}) \approx\mathrm{Poi}(\bar{\bs \lambda}_t)$, $p(\+Z_t| y_{1:t}) \approx\mathrm{Poi}(\bar{\bs \Lambda}_t)$, which we study analytically in Section \ref{sec:theory}, and approximate marginal log-likelihood
\begin{equation}
\log p(y_{1:T}) = \sum_{t=1}^T\log p(y_t|y_{1:t-1}) \approx \sum_{t=1}^T\ell(y_t|y_{1:t-1}),
\end{equation}
which can be used to fit $\mu_q,\sigma^2_q$, and parameters of $\+K_{t,\bs \eta}$, as demonstrated through examples in Section \ref{sec:examples}.
\vspace{-0.5cm}
\section{Asymptotic filtering theory}
\label{sec:theory}
In this section we study the behavior of the approximate filtering quantities and distributions $ p(q_t | y_{1:t}) \approx \mathcal{N}_{[0,1]}(\bar q_t,s_t^2)$, $p(\+x_t| y_{1:t}) \approx\mathrm{Poi}(\bar{\bs \lambda}_t)$, $p(\+Z_t| y_{1:t}) \approx\mathrm{Poi}(\bar{\bs \Lambda}_t)$ in the large population limit on a fixed time horizon $T \in \mathbb{N}_0$. In a change of perspective, for this section we will consider the under-reporting parameters $q_{1:T} \in (0,1)^T$ as fixed but unknown and to be estimated, with \eqref{eq:qprior} defining a prior. The main result of this section, Theorem \ref{thm:1}, rigorously shows that Algorithm \ref{alg:Z} is able to asymptotically recover both the reporting probabilities and the un-observed state of the process. To make dependence on population size explicit, throughout this section we will index the model elements and filtering quantities of algorithm \ref{alg:Z} with the initial population size $n$: $\+x_{t,n}$, $\+Z_{t,n}$, $y_{t,n}$, $\+ \lambda_{t,n}$, $\+ \Lambda_{t,n}$, $\bar q_{t,n},$ and $\bar{\+ \Lambda}_{t,n}$, for $t= 1, \dots, T.$ To be technically complete we write $(\Omega_n, \mathcal{F}_n, \mathbb{P}_n)$ for the probability space underlying the distribution of $\+x_{t,n}$, $\+Z_{t,n}$, $y_{t,n}$ with initial population size $n$, denoting the overall probability space we shall work with as $(\Omega, \mathcal{F}, \mathbb{P}) \coloneqq (\prod_{n \geq 1}\Omega_n, \bigotimes_{n \geq 1}\mathcal{F}_n, \bigotimes_{n \geq 1}\mathbb{P}_n)$. We begin with an assumption.
\begin{assumption}\label{as:1}
There exists a constant $c>0$  such that for all $t \geq 1$, all vectors $\bs f_1, \bs f_2 \in \mathbb{R}^m$, and all probability vectors $\bs \eta, \bs \eta'$:
\begin{equation*}|\bs f_1^\top \+K_{t, \bs \eta}\bs f_2 - \bs f_1^\top \+K_{t, \bs \eta'}\bs f_2 | \leq c \|\bs f_1\|_\infty\|\bs f_2\|_\infty\|\bs \eta - \bs {\eta'}\|_\infty.
\end{equation*}
That is, for all $t\geq 1$, $\+K_{t, \bs \eta}$ is a continuous function of $\bs \eta$.
\end{assumption}

\paragraph{Law of large numbers.}
To analyze the large population limiting properties of the filter, we must first start with the data generating process. Define the following recursions:
\begin{align}
   \bs \nu_0 &\coloneqq \bs \pi_0, \\
    \bs N_t &\coloneqq (\bs \nu_{t-1} \otimes \+1_m)\circ \bs K_{t,\eta(\bs \nu_{t-1})}, \\
   \bs \nu_t &\coloneqq (\+ 1_m^\top \bs N_t)^\top.
\end{align}
Under assumption \ref{as:1} it can be shown \citep{whitehouse2023consistent} that $n^{-1}\+Z_{t,n} \rightarrow \bs N_t$, ${n^{-1}\+x_{t,n} \rightarrow \bs \nu_t}$, and  ${n^{-1}y_{t,n} \rightarrow q_t N_t^{(i,j)}}$ $\mathbb{P}$-almost surely, which can be interpreted as a discrete analogue of the convergence of jump-Markov processes to ODE limits \citep{kurtz1970solutions}. One can interpret $\bs \nu _t$ as the limiting behavior of the (hidden) population prevalence process. 
\paragraph{Asymptotically exact filtering.} Theorem \ref{thm:1} presents a result on the asymptotic properties of algorithm \ref{alg:Z}, in particular it shows that the LawPAL algorithm asymptotically recovers both the prevalence process \textit{and} the under-reporting rates $q_{1:T}$ in the large population limit.

\begin{theorem} \label{thm:1}
    Let assumption \ref{as:1} hold. Then for $t=1,\dots, T$ we have that $n^{-1}\bar{\bs \lambda}_{t,n}\rightarrow \bs \nu_t$, $n^{-1}\bar{\bs \Lambda}_{t,n} \rightarrow \bs N_t$, and $\bar q_{t,n} \rightarrow q_t$ $\mathbb{P}$-almost surely.
\end{theorem}
\begin{proof}
    The proof proceeds by induction on $t$. The result for time $0$ follows directly from the definition of $\bs \lambda_0 = n\bs \pi_0$, hence it is trivial that $\bs n^{-1}\bs \lambda_{0,n} \rightarrow \bs \pi_0 = \bs \nu_0.$   Now, for $t \geq 1$, for our inductive hypothesis we assume that $\bs n^{-1}\bar{\bs \lambda}_{n,t-1} \rightarrow \bs \nu_{t-1}$ $\mathbb{P}$- almost surely. By assumption \ref{as:1} and the continuous mapping theorem we then have that
    \begin{equation}
        n^{-1}\+\Lambda_{t,n} = (n^{-1}\bar{\bs \lambda}_{t-1,n} \otimes \+1_m)\circ\+K_{t,\bs \eta(\bar{\bs \lambda}_{t-1,n})} \rightarrow  (\bs \nu_{t-1} \otimes \+1_m)\circ\+K_{t,\bs \nu_{t-1}} = \bs N_t, \quad \mathbb{P}\text{-a.s.}. 
    \end{equation}
    We now proceed to show that $\bar q_{t,n} \rightarrow q_t$, which requires a more delicate argument. Firstly we write
    \begin{align}
        \bar q_{t,n} &= \arg \max_q \quad y_{t,n}\log\left(q\Lambda_{t,n}^{(i,j)}\right) - q\Lambda_{t,n}^{(i,j)} - \frac{1}{2}\left(\frac{q-\mu_q}{\sigma^2_q} \right)^2+ C_y\\
        &= \arg \max_q \quad\underbrace{n^{-1}y_{t,n}\log\left(\frac{n^{-1}q\Lambda_{t,n}^{(i,j)}}{n^{-1}q_t\Lambda_{t,n}^{(i,j)}}\right) - n^{-1}q\Lambda_{t,n}^{(i,j)}+n^{-1}q_t\Lambda_{t,n}^{(i,j)} - \frac{n^{-1}}{2}\left(\frac{q-\mu_q}{\sigma^2_q} \right)^2}_{%
    \textstyle
    { =:\mathcal{C}_{t,n}(q).}}\\
    \end{align}
    where the second line is obtained from the first by subtracting a quantity constant in $q$ and dividing by $n$. The next step is to notice that by the continuous mapping theorem we have that $\mathbb{P}\text{-a.s.}$
    \begin{align}
        \mathcal{C}_{t,n}(q) &\rightarrow q_tN_t^{(i,j)}\log\left(\frac{q\bs N_{t}^{(i,j)}}{q_t\bs N_{t}^{(i,j)}}\right)-q\bs N_{t}^{(i,j)}+q_t\bs N_{t}^{(i,j)} \\
        &= -\mathrm{KL}\left[\mathrm{Poi}\left(q_t\bs N_{t}^{(i,j)}\right)||\mathrm{Poi}\left(q\bs N_{t}^{(i,j)}\right)\right] =: C_t(q).
    \end{align}
    Hence, by properties of the Kullback Leibler divergence, we have that $\mathcal{C}_{t,n}(q)$ converges $\mathbb{P}$-a.s. to a function maximized (uniquely) at $q_t$. By a standard continuity argument, e.g. \cite[Chapter~12.1]{cappe2005inference}, we then have the maximizer of $\mathcal{C}_{t,n}(q)$ converges $\mathbb{P}$-a.s. to the maximizer of $\mathcal{C}_{t}(q)$ and hence $\mathbb{P}$-a.s. $\bar q_{t,n} \rightarrow q_t$. We can now complete the induction argument, noting that $\mathbb{P}$-a.s
    \begin{equation}
        n^{-1}\bar{\+\Lambda}_{t,n} = n^{-1}y_t + (1-\bar q_{t,n})n^{-1}\+ \Lambda_{t,n} \rightarrow q_t\+N_t + (1 - q_t)\+N_t = \+N_t,
    \end{equation}
    hence by the continuous mapping theorem $n^{-1}\bar{\bs \lambda}_{t,n} =  \left(\+1_m^\top n^{-1} 
  \bar{\+\Lambda}_{t,n}\right)^\top \rightarrow \left(\+1_m^\top \+N_t\right)^\top  = \bs \nu_t$, {$\mathbb{P}$-a.s.}, thus completing the induction.
\end{proof}
If one considers the full approximate filtering distribution $ p(q_t | y_{1:t}) \approx \mathcal{N}_{[0,1]}(\bar q_t,s_t^2)$, we can further notice that since $y_{t,n}\rightarrow \infty$ $\mathbb{P}$-a.s. implies  $s^2_t = \left( \frac{y_{t,n}}{\bar{q}_{t,n}}+ \frac{1}{\sigma_q^2}\right)^{-1} \rightarrow 0$ $\mathbb{P}$-a.s., we have that $\mathcal{N}_{[0,1]}(\bar q_t,s_t^2)$ will converge to a point mass at $q_t$.  To demonstrate this result through simulations we ran a Python implementation of the filter on data from the $SEIR$ model in Section \ref{example:SEIR} for population sizes $n=10^3,10^4,10^5$ and $T=100$. We first simulated $q_{1:T}\sim \mathcal{N}_{[0,1]}(0.5,0.1)$ which are then fixed for all values of $n$, for each population size we then simulate $y_{1:T}$ using parameters $\beta = 0.8,\rho = 0.1,\gamma = 0.2, \bs\pi_0=[0.99,0,0.1,0]$ and observation model \ref{sec:obsmod}. For each simulation corresponding to each value of $n$ we ran the LawPAL filter. Figure \ref{fig:filtering} plots the approximate filtering means and 95\% credible bands for $q_t$ against the ground truth values, demonstrating that, for each $t$, $\mathcal{N}_{[0,1]}(\bar q_t,s_t^2)$ collapses onto the true value $q_t$ as per Theorem \ref{thm:1}. In the next section we assess the ability of LawPAL approximate likelihood based inference to recover ground truth $\mu_q,\sigma^2_q$, and parameters of $\+K_{t,\bs \eta}$.
\begin{figure}
    \centering
    \includegraphics[width=1\linewidth]{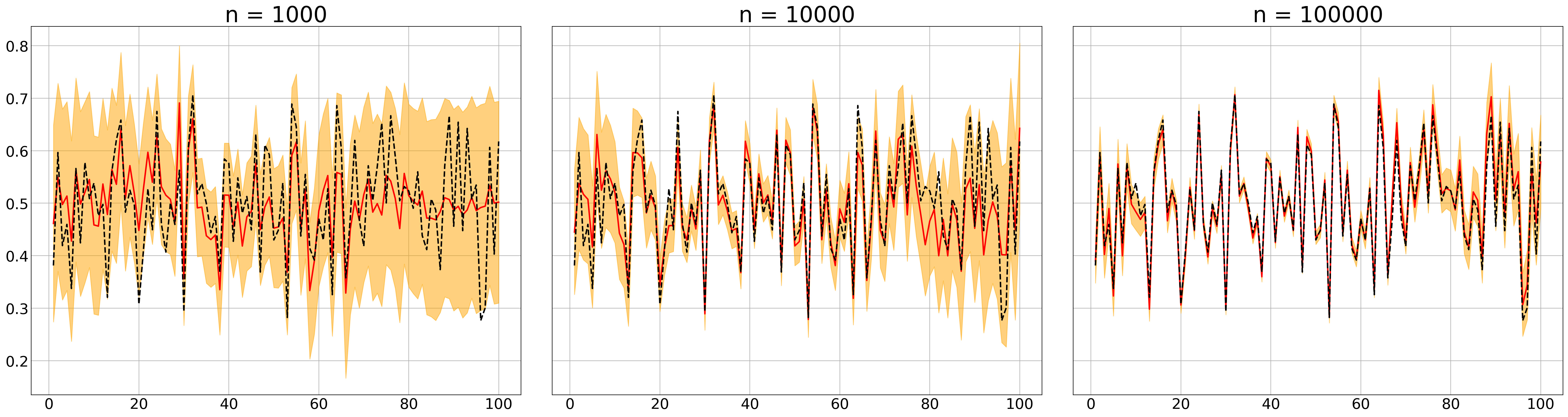}
    \caption{LawPAL filtering distributions for the $SEIR$ example, dotted lines indicate ground truth $q_{1:T}$, red lines the LawPAL filter mean, orange bands LawPAL filter 95\% credible region.}
    \label{fig:filtering}
\end{figure}
\vspace{-0.5cm}
\section{Experiments and examples} \label{sec:examples}
Code for all examples can be found at github.com/Michael-Whitehouse/LawPAL, implemented using Rcpp.
\subsection{Parameter recovery} \label{sec:paramrec}
Section \ref{sec:theory} studies the theoretical performance of the filter under a well specified parameterization. Whilst in some settings it is reasonable to make assumptions on epidemiological parameters \citep{kucharski2020early}, in practice these quantities are unknown and must be estimated from data. Here we study the behavior of the estimator obtained by maximizing the approximate likelihood of algorithm 1, as population size $n$ and time horizon $T$ \emph{both} grow large.

We consider a Susceptible-Infected-Removed $(SIR)$ type model defined by an initial distribution $\+x_0\sim \mathrm{Mult}(n,\bs\pi_0)$, for some length-$4$ probability vector $\bs\pi_0$, and transition matrix

\begin{equation}\label{eq:SIRexample}
\mathbf{K}_{t, \boldsymbol{\eta}}=\left[\begin{array}{ccc}
e^{-h \beta \eta^{(2)}} & 1-e^{-h \beta \eta^{(2)}} & 0 \\
0 &e^{-h \gamma} & 1-e^{-h \gamma}\\
0  &0 &1
\end{array}\right].
\end{equation}
Observations are associated with new infections $y_t\sim \mathrm{Binom}(Z_t^{1,2},q_t)$ where $q_t\sim\mathcal{N}_{[0,1]}(\mu_q,\sigma^2_q)$, for some $\mu_q \in [0,1], \sigma^2_q>0$. We simulate from this model using initial distribution $\bs\pi_0 = [0.995,0.005,0]$ and parameters $\beta =0.15,\gamma=0.1,\mu_q=0.5,\sigma^2_q=0.1$, repeated with varying population sizes $n=5000,10^4,10^5,10^6$ and time horizons $T=50,100,150,200$, and consider the problem of jointly estimating $\beta,\gamma,\mu_q,\sigma^2_q$ with $\bs\pi_0$ assumed known.
Approximate maximum likelihood estimates are obtained by maximizing $\ell(y_{1:T})$ jointly with respect to $\beta,\gamma,\mu_q,\sigma^2_q$ using a simple coordinate ascent algorithm, detailed in the associated github repository. For each pair of $n$ and $T$ we produce $100$ trajectories $y_{1:T}$ and $100$ corresponding approximate maximum likelihood estimators of $\beta,\gamma,\mu_q,\sigma^2_q$, summaries of these estimates are presented in tables \ref{tab:consistency}(a)-(d).

\begin{table}[htb]
    \centering
    \begin{subtable}[b]{0.495\textwidth} 
        \centering
        {\scriptsize
        \begin{tabular}{c | c c c c} 
        n &$5000$&$10^4$&$10^5$&$10^6$  \\
            \hline
           $\beta$ & $0.165(.040)$ & $0.155(.032)$ & $0.154(.019)$ & $0.153(.020)$ \\ \hline
            $\gamma$ & $0.117(.043)$ & $0.106(.035)$ & $0.104(0.019)$ & $0.103(.021)$\\ \hline
            $\mu_q$ & $0.484(.080)$ & $0.480(.077)$ & $0.493(.088)$ & $0.503(.091)$\\ \hline
            $\sigma^2_q$ &$0.092(.106)$ & $0.097(.066)$ & $0.096(.040)$ & $0.097(.039)$ \\ \hline
        \end{tabular}
        \caption{\footnotesize $T=50$}
        }
    \end{subtable}
    \hfill
    \begin{subtable}[b]{0.495\textwidth}
        \centering
        {\scriptsize
        \begin{tabular}{c | c c c c}
        n &$5000$&$10^4$&$10^5$&$10^6$  \\
            \hline
            $\beta$ &  $0.144(.030)$ & $0.147(.017)$ & $0.149(.006)$ & $0.150(.003)$  \\ \hline
            $\gamma$ & $0.094(0.30)$ & $0.99(.019)$ & $0.99(.006)$ & $0.010(.004)$ \\ \hline
            $\mu_q$ & $0.493(.069)$ & $0.507(.052)$ & $0.499(.022)$ & $0.498(.020)$ \\ \hline
            $\sigma^2_q$ &$0.105(.066)$ & $0.104(.036)$ & $0.099(.021)$ & $0.098(.014)$ \\ \hline
        \end{tabular}
        \caption{\footnotesize $T=100$}
        }
    \end{subtable}

    \vspace{0.1cm}

    \begin{subtable}[b]{0.495\textwidth}
        \centering
        {\scriptsize
        \begin{tabular}{c | c c c c}
        n &$5000$&$10^4$&$10^5$&$10^6$  \\
            \hline
            $\beta$ & $0.151(.021)$ & $0.147(.016)$ & $0.150(.005)$ & $0.150(.002)$  \\ \hline
            $\gamma$ & $0.102(.024)$ & $0.097(0.018)$ & $0.099(.006)$ & $0.010(.002)$ \\ \hline
            $\mu_q$ & $0.510(.064)$ & $0.500(.049)$ & $0.498(.017)$ & $0.501(.009)$ \\ \hline
            $\sigma^2_q$ &$0.107(.065)$ & $0.101(.036)$ & $0.099(.013)$ & $0.099(.011)$\\ \hline
        \end{tabular}
        \caption{\footnotesize $T=150$}
        }
    \end{subtable}
    \hfill
    \begin{subtable}[b]{0.495\textwidth}
        \centering
        {\scriptsize
        \begin{tabular}{c | c c c c}
        n &$5000$&$10^4$&$10^5$&$10^6$  \\
            \hline
            $\beta$ & $0.151(.017)$ & $0.151(.015)$ & $0.149(.004)$ & $0.150(.001)$  \\ \hline
            $\gamma$ & $0.102(.021)$ & $0.103(.018)$ & $0.100(.004)$ & $0.100(0.001)$ \\ \hline
             $\mu_q$ & $0.508(.064)$ & $0.511(.051)$ & $0.500(.013)$ & $0.500(.008)$ \\ \hline
            $\sigma^2_q$ &$0.112(.067)$ & $0.106(.040)$ & $0.101(.015)$ & $0.100(.010)$\\ \hline
        \end{tabular}
        \caption{\footnotesize $T=200$}
        }
    \end{subtable}
    \caption{Approximate maximum likelihood estimate sample means and standard deviations.}
    \label{tab:consistency}
\end{table}
It is clear that as $n$ and $T$ increase the estimate samples become more concentrated around ground truth values; demonstrating that the LawPAL approximate likelihood can be used to recover ground truth parameters. In the large population $n\rightarrow \infty$ regime with a fixed time horizon $T<\infty$, a consistency result along with identifiability conditions are rigorously established in \cite{whitehouse2023consistent} for an equi-dispersed model, corresponding to an atomic distribution for the reporting probability. Tables \ref{tab:consistency}(a)-(d) provide evidence that a consistency result may be feasible for the over-dispersed auxiliary variable model in the \emph{joint} $n\rightarrow\infty$, $T\rightarrow\infty$ regime. An intuition for the need to also take $T\rightarrow\infty$ comes by considering that trajectories $y_{1:T}$ with a fixed $T$ can only contain information related to a finite sample $q_{1:T}\sim\mathcal{N}_{[0,1]}(\mu_q,\sigma^2_q)$, and hence finite information related to $\mu_q$ and $\sigma_q^2$; for more information we would need more samples and hence a larger $T$. This heuristic also explains the variance of the $\mu_q$ estimate in table \ref{tab:consistency}(a) does not decrease as $n$ increases: without $T$ also increasing, little information on $\mu_q$ is gained. A rigorous proof for such a consistency result and identifiability conditions in the joint regime would require significant development of the approach and additional tools to those used in \cite{whitehouse2023consistent}, e.g. time uniform bounds on convergence rates of the data generating process. Such a result is an enticing prospect, but out of the scope of the present paper; we leave its study to future work.

\subsection{Computational efficiency and statistical compromise} \label{sec:smccomp}
Section \ref{sec:paramrec} demonstrates the ability of the approximate maximum likelihood estimator to recover data generating parameters. In this section we characterize the statistical price paid for the computational efficiency gained when using the LawPAL instead of sequential Monte Carlo.
\subsubsection{Comparison with sequential Monte Carlo}
It is well known that SMC can be used to consistently (as the number of particles $\rightarrow\infty$) estimate the true likelihood for a hidden Markov model \citep{chopin2020introduction}, furthermore, using SMC within a pseudo marginal Metropolis Hastings sampler produces an ergodic chain targeting the true posterior distribution \citep{andrieu2010particle}. As such, SMC is often described as a `gold standard' approximation and provides us with a way to check the performance of the LawPAL algorithm, with the heuristic being if our approximations closely match those obtained with SMC, then our approach well approximates the true likelihood. 
\begin{figure}[h] 
    \centering
    \includegraphics[width=1\textwidth]{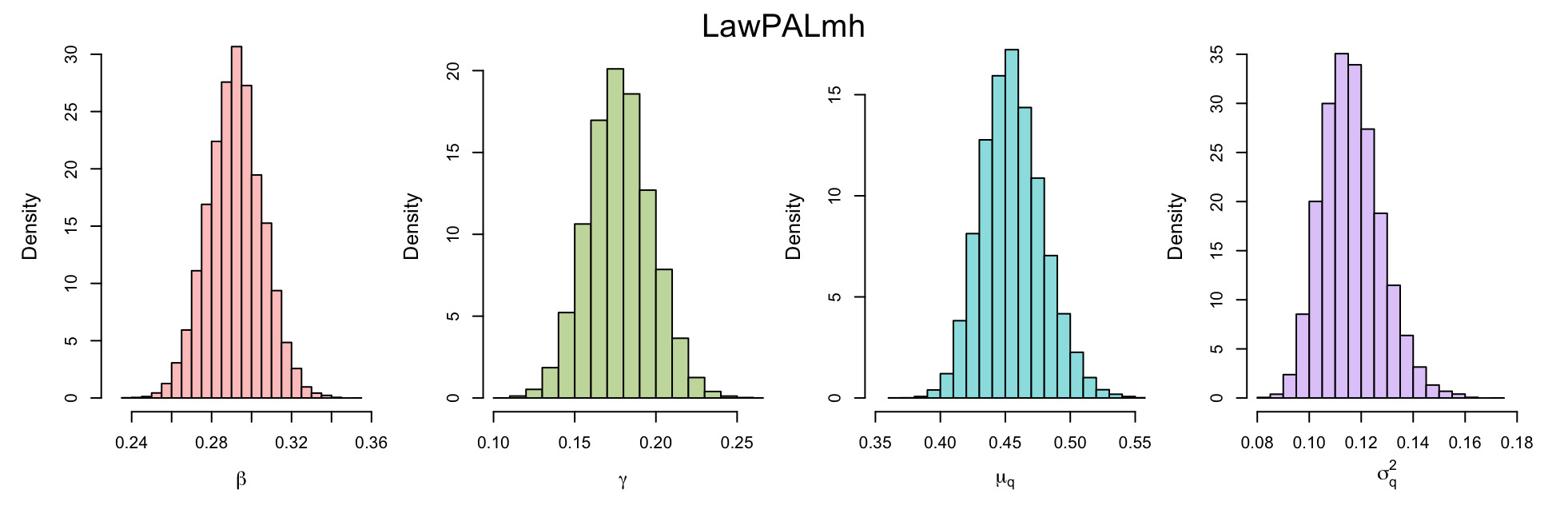}
    \vfill  
    
    \includegraphics[width=1\textwidth]{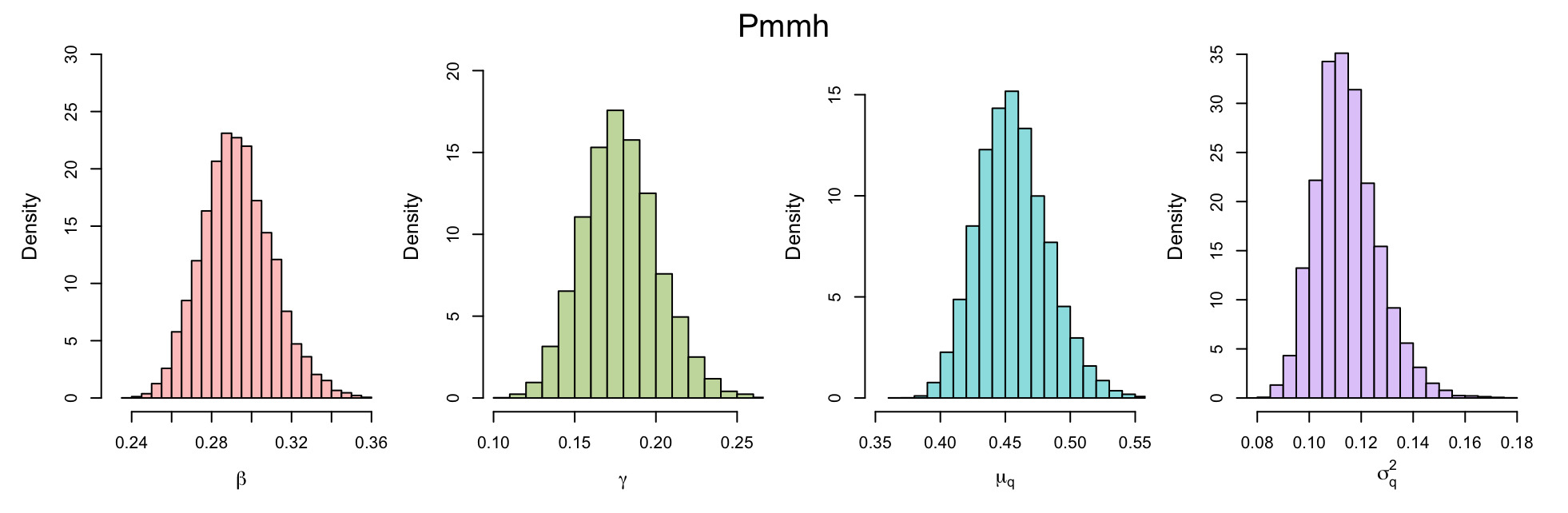}
    \caption{Metropolis Hastings plots for the $n=2.5\times10^4$, $T=50$ simulated data.} \label{tab:moderate}
\end{figure}
To this end, we simulate data  from model \ref{eq:SIRexample} with $\bs\pi_0 = [0.995,0.005,0]$ and parameters $\beta =0.3,\gamma=0.2,\mu_q=0.5,\sigma^2_q=0.1$ under 2 regimes, once with  $n=2.5\times10^4$, $T=50$ (of  order corresponding a moderately sized town) and once with $10^6$, $T=100$ (of order corresponding a large city). For each of $\beta,\gamma,\sigma_q^2$ we place a vague $\mathcal{N}_{[0,\infty]}(0,10)$ prior, for $\mu_q$ we use $\mathcal{N}_{[0,1]}(0.5,10)$. For each of the simulations we ran two Metropolis Hasting schemes, 1) with the likelihood approximated by SMC using $1000$ particles (Pmmh), 2) with the likelihood approximated by the LawPAL algorithm (LawPALmh). Priors are kept consistent for both schemes so that any discrepancies in posterior distributions must be due to discrepancies in likelihood approximations. For each chain we ran a $2000$ iteration burn in with independent Gaussian random walk proposals with variance $0.01$, then use this to construct an `optimal' \citep{andrieu2008tutorial} joint multi-variate Gaussian proposal covariance $\Sigma = (2.38^2/d)\hat\Sigma$ where $\hat\Sigma$ is the sample covariance of the burn in chain, and use this to produce a final chain of $10^5$ iterations. Each chain showed no signs of poor mixing and exhibited satisfactory decay in auto-correlation plots. The LawPALmh scheme performed $10^4$ iterations 3.8 seconds, the Pmmh scheme took 345 seconds, demonstrating a $\sim 90\times$ speed up. We further note that for the pmmh scheme specially designed proposal distributions \citep{whitehouse2023consistent} were required to reduce the marginal likelihood estimate variance to ensure adequate estimation of the likelihood ratio step in the Metropolis Hastings sampler; this extra step was not necessary for the deterministic LawPAL marginal likelihood approach.

 Figures \ref{tab:moderate} and \ref{tab:large} plot posterior histogram densities for each regime using a thinned sample of size $3\times 10^4$. Under the moderate population regime the statistical price of the approximation becomes apparent: whilst both schemes produce posterior distributions concentrated close to the data generating parameters, there are disagreements in the tails of the distributions with the pmmh posteriors exhibiting higher variance, indicating presence of slight over-confidence in LawPALmh posteriors. Under the large population regime this discrepancy becomes less pronounced, as exhibited by figure \ref{tab:large}. This can be heuristically linked Theorem 1: as $n$ grows large, the first moment $\bar q_t$ becomes more representative of the full filtering distribution as it collapses to a single point. Hence, less is sacrificed when we propagate the point value $\bar q_t$ in algorithm line 7, as opposed to the full conditional distribution.

\begin{figure}[t] 
    \centering
    \includegraphics[width=1\textwidth]{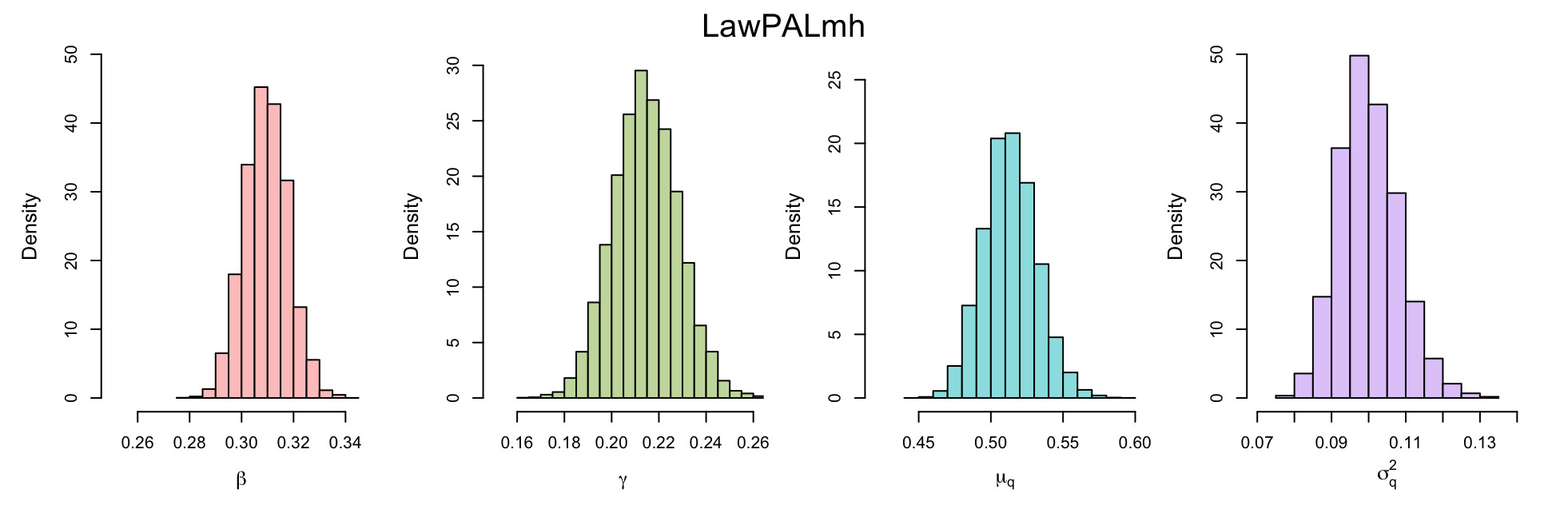}
    \vfill  
    
    \includegraphics[width=1\textwidth]{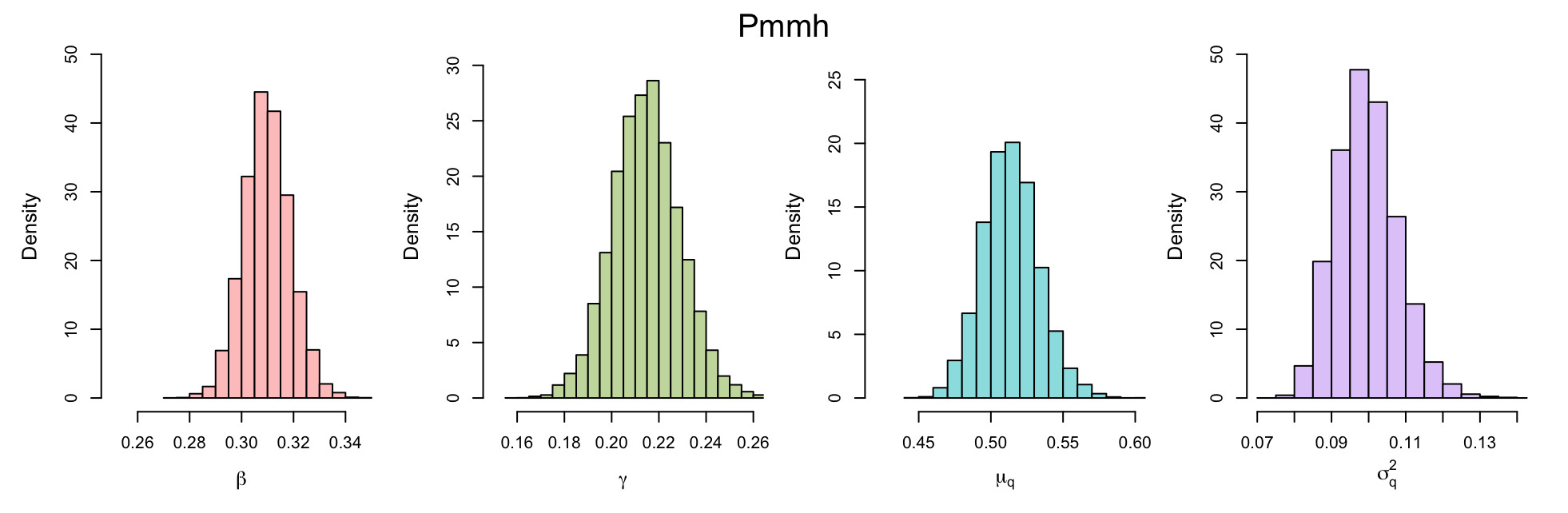}
    \caption{Metropolis Hastings plots for the $n=10^6$, $T=100$ simulated data.}
    \label{tab:large}
\end{figure}
\subsection{Bayesian inference with Stan: Covid-19 in Switzerland} \label{sec:covid}
In this section we demonstrate how the LawPAL algorithm can be used for the development of a model for real data of practical interest using `off the shelf' software. Probabilistic programming languages (PPL) have recently been recognized as a valuable tool in the development of disease transmission models \citep{grinsztajn2021bayesian, andrade2020evaluation}. In contrast to sampling based methods, the LawPAL approximation comprises only simple linear algebraic operations and is therefore naturally amenable to auto-differentiation and use within PPL libraries to draw posterior samples using an Hamiltonian Monte Carlo (HMC) scheme. HMC samplers are a type of MCMC algorithm which use auxiliary ‘momentum’ variables to aid in exploring the posterior, crucially they require access to the gradient of the likelihood function to sample Hamiltonian trajectories. 

We consider data of Covid-19 incidence reported in Switzerland in 2020. The count time-series data we consider is freely available, e.g. from \cite{grinsztajn2021bayesian}, and starts from the index case lasting 109 time steps. Our transmission model is a discrete-time stochastic counterpart to the ODE SEIR 
 model developed by \cite{grinsztajn2021bayesian}, which captures the effects of a national control measures implemented at time $t^*=23$ using a time varying transmission rate.
 \subsubsection{SEIR model with control intervention}
Cast as an instance of our latent compartmental model, the transmission process is governed by the time dependent transition matrix
 \begin{equation}\label{eq:SEIRcovid}
\mathbf{K}_{t, \boldsymbol{\eta}}=\left[\begin{array}{cccc}
e^{- \beta_t \eta^{(3)}} & 1-e^{-h \beta_t \eta^{(3)}} & 0 &0\\
0 &  e^{- \rho} & 1-e^{- \rho} &0\\
0 &  0&e^{- \gamma} & 1-e^{- \gamma}\\
0 & 0  &0 &1
\end{array}\right],
\end{equation}
with $\beta_t = \beta(\alpha +(1-\alpha)(1+\exp(b(t-t^*-d)))^{-1})$, where $\beta$ denotes the uncontrolled transmission rate, $\alpha$ is the decrease in transmission when controls are fully in place, $b$ is the slope of the decrease, and $d$ is the delay in the effect of control measures. Observations are $y_t\sim\mathrm{Binom}(Z_t^{(2,3)},q_t)$ where we consider two distinct models: an \textbf{equi-dispersed} observation model $q_t = \mu_q$ for all $t$, and an \textbf{over-dispersed} observation model $q_t \sim \mathcal{N}_{[0,1]}(\mu_q,\sigma_q^2)$, where $\mu_q\in [0,1]$ and $\sigma^2_q>0$. The population size is taken to be $n=8.57\times 10^6$ and the initial distribution probability vector is $\bs \pi_0 = n^{-1} [n-i_0-e_0,i_0,e_0,0]$ where $i_0,e_0>0$ are parameters.


\begin{table}[ht]
\centering
{\footnotesize
\begin{tabular}{c|c c c}
Parameter & Prior & over-dispersed  & equi-dispersed \\
\hline
$\beta$ & $\mathcal{N}_{[0,\infty)}(2,0.5)$ & $1.53(0.33)$ & $1.24(0.04)$ \\
$\rho$ & $\mathcal{N}_{[0,1]}(0.2,0.1)$ & $0.17(0.06)$ & $0.07(0.0007)$ \\
$\gamma$ & $\mathcal{N}_{[0,1]}(0.2,0.1)$ & $0.33(0.08)$ & $0.003(0.004)$ \\
$\alpha$ & ${Beta}(2.5,4)$ & $0.09(0.04)$ & $0.0005(0.00003)$ \\
$b$ & ${Beta}(1,1)$ & $0.24(0.24)$ & $0.013(0.001)$ \\
$d$ & $Exp(0.1)$ & $3.31(0.77)$ & $1.54(0.12)$\\
$\mu_q$ & ${Beta}(1,2)$ & $0.62(0.17)$ & $0.67(0.09)$ \\
$i_0$ & $\mathcal{N}_{[0,\infty)}(0,20)$ & $24.5(9.49)$ & $41.2(6.30)$ \\
$e_0$ & $\mathcal{N}_{[0,\infty)}(0,20)$ & $15.6(10.31)$ & $12.1(6.42)$ \\
$\sigma_q^2$ & ${Exp}(0.1)$ & $0.21(0.06)$ & NA \\
\end{tabular}
}
\caption{Parameter priors and posterior mean (standard deviation) summaries.}
\label{tab:swissparams}
\end{table}
\vspace{-0.5cm}
\subsubsection{Inference with HMC} \label{sec:hmcsampler}
The PPL Stan \citep{carpenter2017stan} provides a framework for the implementation of HMC in which a user only needs to provide priors and a likelihood function to produce posterior samples, providing a far less labor intensive approach than the pmmh scheme explored in Section \ref{sec:smccomp}.
The equi-dispersed model is fit by embedding the Poisson Approximate Likelihood (PAL) algorithm \citep{whitehouse2023consistent} in a Stan program, the over-dispersed model is similarly fit by embedding the LawPAL algorithm, implementation details can be found on the github page. For both models the parameters to be learned are $\beta,\rho,\gamma,\alpha,b,d,\mu_q,i_0,e_0$ with the over-dispersed model having the additional dispersion parameter $\sigma_q^2$. Vague priors, inspired by \cite{grinsztajn2021bayesian}, are assigned to each of these parameters and are given in table \ref{tab:swissparams}. For each model the HMC sampler was run to produce chains of length $10^5$ which showed no signs of unsatisfactory mixing. Posterior summaries are presented in table \ref{tab:swissparams}, there are significant disagreements between our two models, which we explore in the next section.
\subsubsection{Posterior predictive checks: model parsimony}

\begin{figure}[h]
    \centering
    \begin{subfigure}[b]{0.495\textwidth}
        \centering
        \includegraphics[width=\linewidth]{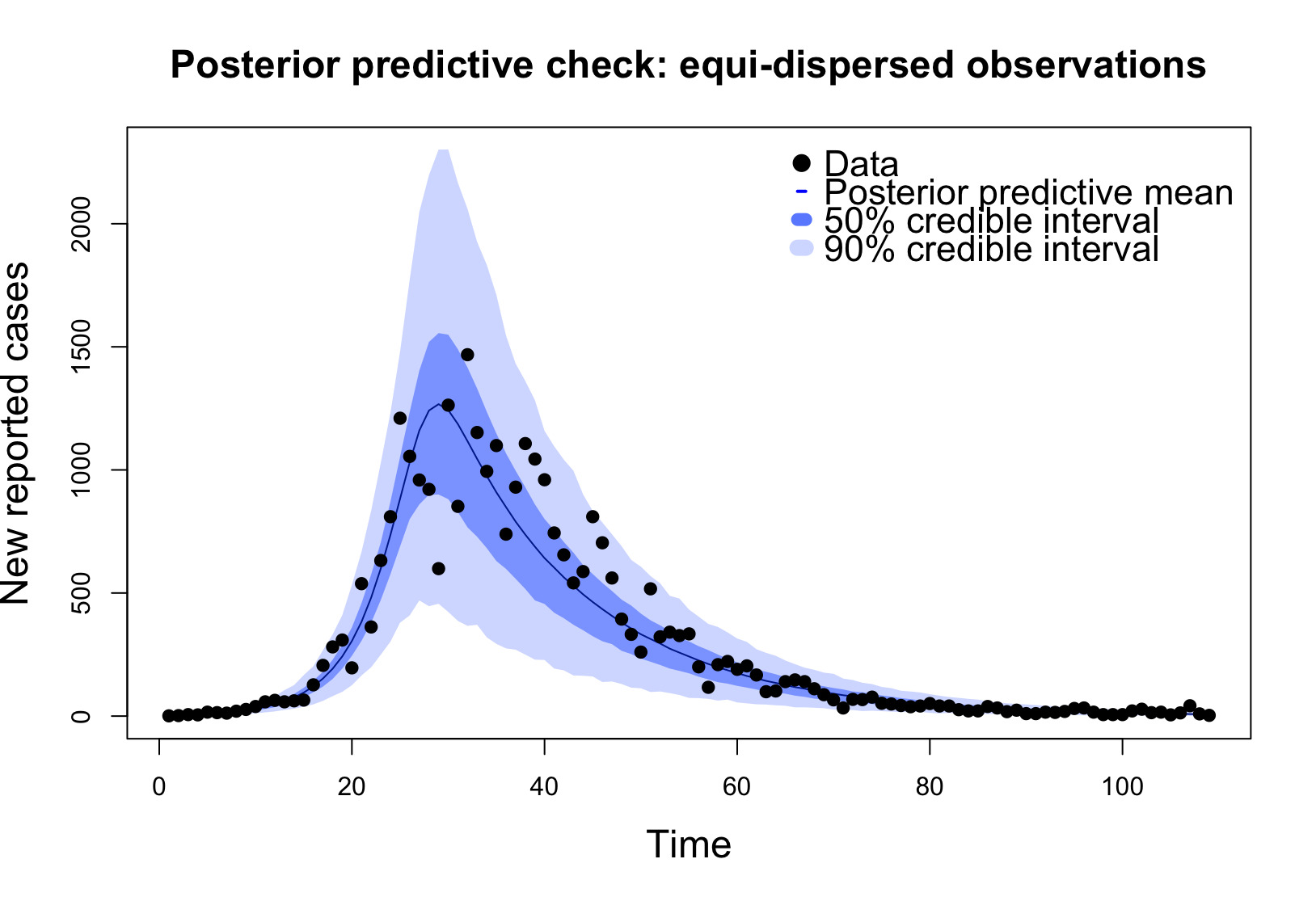}
        \label{fig:plot1}
    \end{subfigure}
    \hfill
    \begin{subfigure}[b]{0.495\textwidth}
        \centering
        \includegraphics[width=\linewidth]{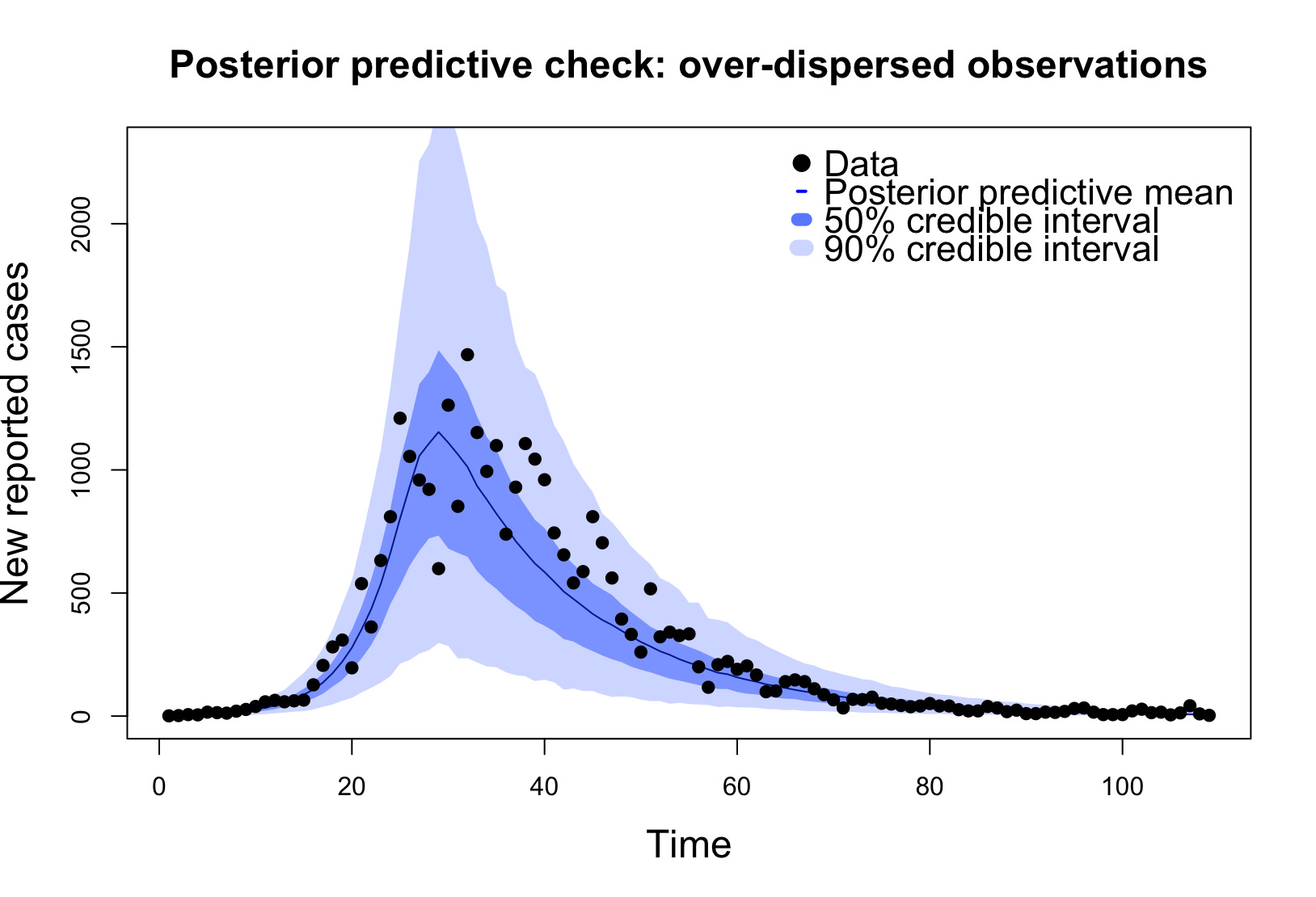}
        \label{fig:plot2}
    \end{subfigure}
    \vspace{-1.5cm}
    \caption{Graphical posterior predictive distributions with credible intervals.}
    \label{fig:postpredband}
\end{figure}

\begin{figure}[h]
    \centering
    \begin{subfigure}[b]{0.495\textwidth}
        \centering
        \includegraphics[width=\linewidth]{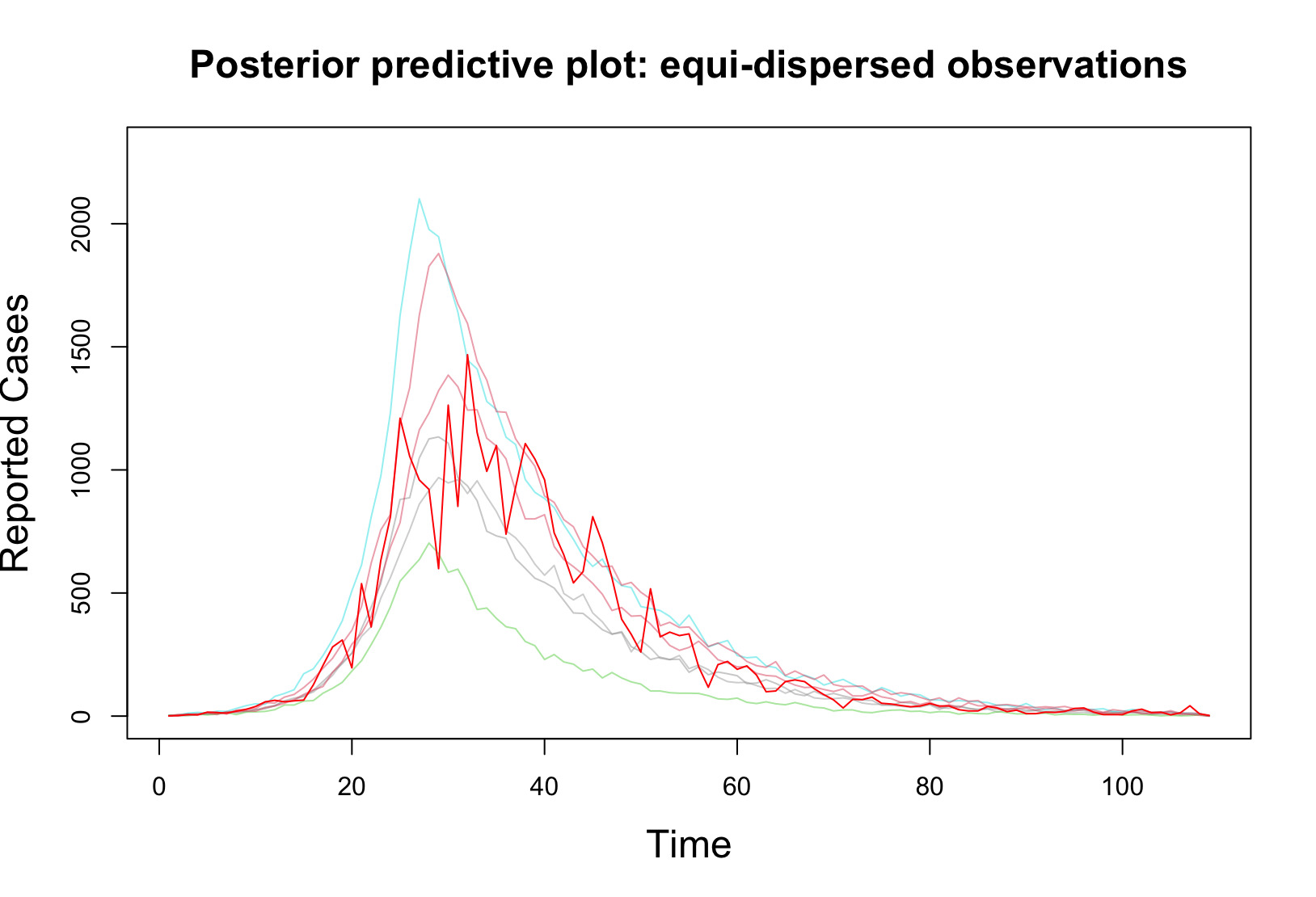}
        \label{fig:plot1}
    \end{subfigure}
    \hfill
    \begin{subfigure}[b]{0.495\textwidth}
        \centering
        \includegraphics[width=\linewidth]{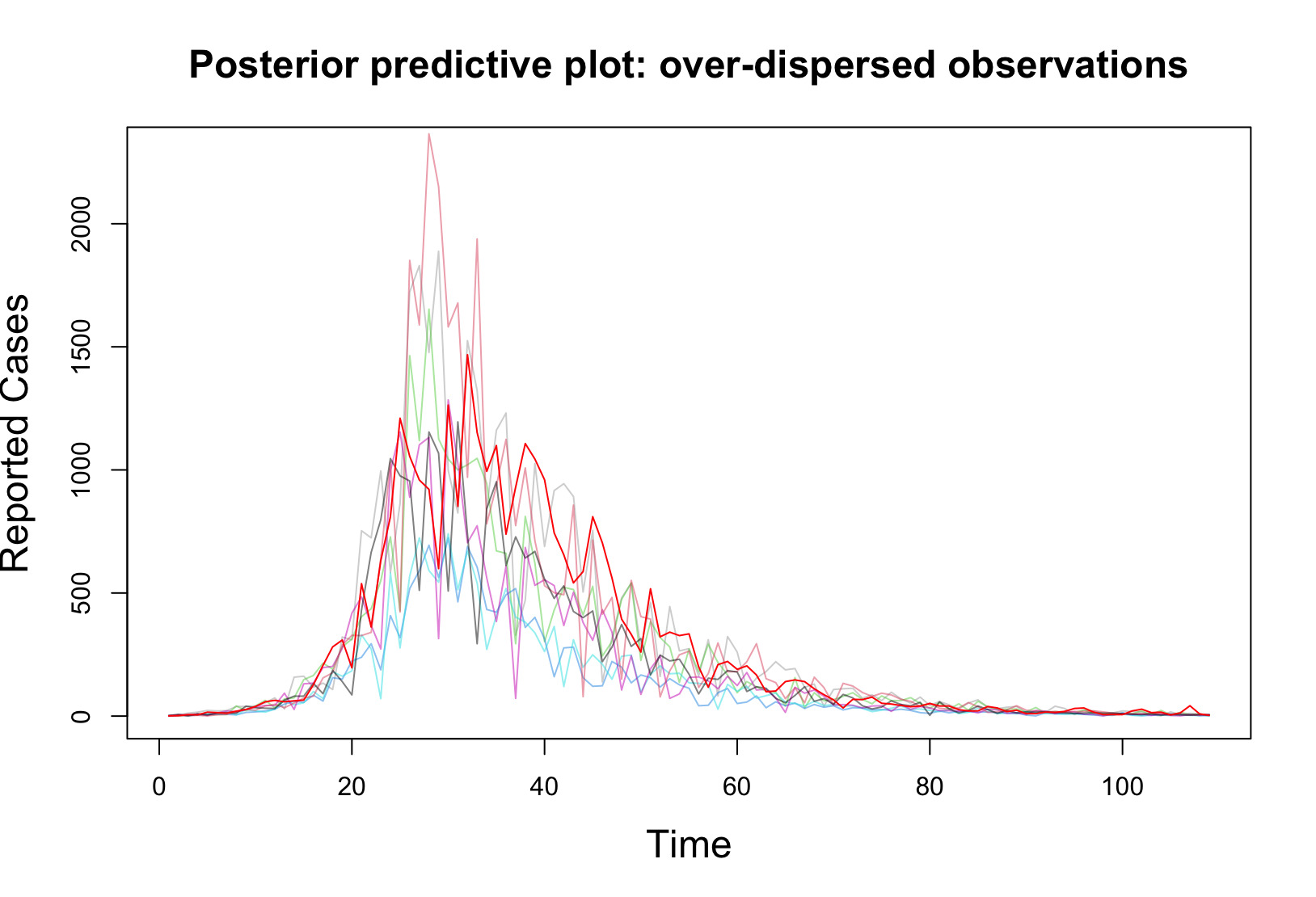}
        \label{fig:plot2}
    \end{subfigure}
    \vspace{-1.5cm}
    \caption{Graphical posterior predictive plots. Solid red line plots observed data. Coloured faint lines plot simulations from the posterior predictive distribution.}
    \label{fig:spaghetti}
\end{figure}

Graphical model checks provide a way for practitioners to look for systematic discrepancies between observations and model simulations by plotting them alongside one another and making comparisons \cite[Chapter 6.4]{gelman1995bayesian}. We now apply this philosophy to assess our two models.
Using the posterior samples obtained in Section \ref{sec:hmcsampler} we sampled from the posterior predictive distributions for the equi-dispersed and over-dispersed models. In figure \ref{fig:postpredband} we plot each of these distributions with credible intervals. Based on these plots alone both models appear reasonable, each with strong coverage of the observed data. However, these plots are effectively 2-dimensional projections of $109$ dimensional objects, summarizing only marginal distributions for each time step and obscuring between-timestep covariation. Figure \ref{fig:spaghetti} presents `spaghetti' plots of posterior predictive draws alongside the observed data. Using these graphical checks the limitations of the equi-dispersed model are exposed, with the observed data exhibiting significantly more between-timestep variation than the posterior draws. This is in contrast to the corresponding plot for the over-dispersed model, in which posterior predictive draws and the observed trajectory are indistinguishable in this regard. Furthermore the posterior samples associated with the equi-dispersed model imply an astronomically high posterior mean $R_0=\beta/\gamma \approx 8872$, which is epidemiologically absurd. The corresponding estimate for the over-dispersed model is a far more reasonable $R_0 \approx 4.4$. These issues exhibit the potential pitfalls of assuming equi-dispersed observation models, and demonstrate how the LawPAL has allowed us to build a model parsimonious with observed data using an `off the shelf' sampler, in a way that the vanilla PAL algorithm does not permit. Furthermore the exercise exhibits the need for a holistic approach to posterior evaluation and to inform model selection decisions with multiple plots and checks.

\vspace{-0.5cm}
\section{Discussion and limitations} \label{sec:discussion}
We have developed an analytically tractable ADAL procedure for stochastic epidemic models permitting observational over-dispersion, supported with theoretical results and favorable experimental performance. This extends the remit of ADAL methods to a broader class of models than previously possible and provides order of magnitude computational advantages over simulation-based sequential Monte Carlo methods. A natural extension of this work would be to tackle models with mechanistic over-dispersion \citep{ning2024systemic,whitehouse2023consistent}, e.g. by treating transmission rates as latent variables. It could be fruitful to investigate ADAL style approximations for such models, alternatively the LawPAL algorithm could be embedded within an SMC scheme. Models with increasing numbers of auxiliary variables inducing increasing levels of over-dispersion define a set of nested models, each with a hierarchical latent structure. It would be interesting to compare such models using Bayesian selection techniques which explicitly penalize model variability pertinent to hierarchical models with latent structure, such as the Watanabe-Akaike information criterion or deviance information criterion \citep{watanabe2013widely,spiegelhalter2002bayesian}. ADAL methods provide a taxonomy for accelerating posterior computations and hence could be used to develop a Bayesian workflow \citep{gelman1995bayesian} for the design and selection of models with differing extents of over-dispersion, as a Bayesian compliment to a previously explored frequentist approach \citep{stocks2020model}. The approximations developed in this work are derived from the specifics of the model and subsequently are not as general as simulation based methods such as particle filtering, highlighting a trade off between model choices and computational cost. We envisage that the use of ADAL methods in conjunction with particle filter methods has the potential to provide significant aid in scaling up procedures to large systems by reducing the dimensionality of the filtering problem, all without compromising the fidelity of the inference, as previously explored by \cite{whitehouse2023consistent} in the context of look-ahead filtering schemes \citep{rebeschini2015can}. Future work could also aim to provide a generalized recipe for approximation step 2 through the use of auto-differentiation software, further investigation would be needed to elucidate the potential gains of such an approach - taking inspiration from integrated nested Laplace approximation \citep{rue2009approximate} type approaches could be of benefit to this end.


\textbf{Funding:}  MW acknowledges funding from the MRC Centre for Global Infectious Disease Analysis (Reference No. MR/X020258/1) funded by the UK Medical Research Council.

\textbf{Declaration:} The authors report there are no competing interests to declare.


    
\vspace{-0.5cm}
\bibliographystyle{chicago}
\bibliography{bibliography}

\end{document}